\documentclass[english]{article}
\usepackage{graphicx}
\usepackage{color}
\usepackage[ruled]{algorithm2e}
\usepackage{comment}
\usepackage{lmodern}
\usepackage[T1]{fontenc}
\usepackage[a4paper]{geometry}
\usepackage[authoryear]{natbib}
\usepackage{babel}
\usepackage{float}
\usepackage{enumitem}
\usepackage{amsmath}
\usepackage{amsthm}
\usepackage{amssymb}
\usepackage{bbm}
\usepackage{appendix}
\usepackage[unicode=true,
 bookmarks=true,bookmarksnumbered=false,bookmarksopen=false,
 breaklinks=true,pdfborder={0 0 0},pdfborderstyle={},colorlinks=true]
 {hyperref}
\usepackage{cleveref}
\usepackage{xspace}
\usepackage{dsfont}
\usepackage[dvipsnames,svgnames,x11names,hyperref]{xcolor}
\hypersetup{colorlinks,linkcolor={RoyalBlue},urlcolor=RoyalBlue,citecolor=RoyalBlue}
\usepackage[normalem]{ulem}
\usepackage{caption}
\usepackage{subfig}
\usepackage{authblk}

\newcommand{\CB}{\ensuremath{\text{CB}}}

\newcommand{\diagset}{\bar{\mathbb{X}}_\mathrm{D}}
\newcommand{\favset}{\bar{\mathbb{X}}_\mathrm{F}}
\newcommand{\unfavset}{\bar{\mathbb{X}}_\mathrm{U}}
\newcommand{\adjset}{\bar{\mathbb{X}}_{adj}}
\newcommand{\xifavtomeet}{\xi_{\mathrm{F}\to\mathrm{D}}}
\newcommand{\xiunfavtofav}{\xi_{\mathrm{U}\to\mathrm{F}}}
\newcommand{\xifavtounfav}{\xi_{\mathrm{F}\to\mathrm{U}}}
\newtheorem{theorem}{Theorem}
\theoremstyle{definition}

\newtheorem{assumption}{Assumption}[section]
\newtheorem{proposition}{Proposition}[section]

\theoremstyle{remark}

\title{A simple Markov chain for independent Bernoulli variables \\ conditioned on their sum}
\author[1]{Jeremy Heng}
\author[2]{Pierre E. Jacob}
\author[2]{Nianqiao Ju \thanks{Corresponding author: nju@g.harvard.edu}}
\date{}
\affil[1]{ESSEC Business School, Singapore}
\affil[2]{Department of Statistics, Harvard University, USA}

\begin{document}
\maketitle

\begin{abstract}
  We consider a vector of $N$ independent binary variables, each
  with a different probability of success.  The distribution of the vector
  conditional on its sum is known as the conditional
  Bernoulli distribution. Assuming that $N$ goes to infinity
  and that the sum is proportional to $N$, exact sampling 
  costs order $N^2$, while a simple Markov chain Monte Carlo
  algorithm using ``swaps'' has constant cost per iteration.  
  We provide conditions under which this Markov chain converges in order $N \log N$ iterations.
  Our proof relies on couplings and an auxiliary Markov chain defined on 
  a partition of the space into favorable and
  unfavorable pairs. 
\end{abstract}

\section{Sampling from the conditional Bernoulli distribution}

\subsection{Problem statement}

Let $x=(x_{1},\ldots,x{}_{N})$ be an $N$-vector in $\{0,1\}^{N}$, with sum
$\sum_{n=1}^{N}x_{n}=I$. Let $(p_{1},\ldots,p_{N}) \in (0,1)^N$,
 and denote the associated ``odds'' by $w_{n}=p_{n}/(1-p_{n})$.
Define the set $S_{z}=\{n\in[N]:x_{n}=z\}$ for $z\in\{0,1\}$, where $[N]=\{1,\ldots,N\}$,
i.e. $S_z$ has indices $n\in[N]$ at which $x_n=z$. We consider the
task of sampling $x\in\{0,1\}^N$ from a distribution obtained by specifying an 
independent Bernoulli distribution with probability $p_n$ 
on each component $x_n$, and conditioning on $\sum_{n=1}^{N}x_{n}
= I$ for some value $0\leq I\leq N$. This is known as the conditional Bernoulli
distribution and will be denoted by $\CB(p,I)$. 
The support of
$\CB(p,I)$ is denoted as $\mathbb{X}=\{x\in\{0,1\}^N: \sum_{n=1}^N x_n = I\}$. 
We assume that $1\leq I \leq N/2$, 
since 
we can always swap the labels ``0''and ``1''. 
We consider the asymptotic regime where $N$ and $I$ go to infinity at the same rate.

One can sample exactly from $\CB(p,I)$
\citep{chen1994weighted,chen1997statistical}, for a cost of order $IN$, 
thus order $N^2$ in the context of interest here; see Appendix \ref{appx:exactsampling}. 
As an alternative, we consider a simple Markov chain Monte 
Carlo (MCMC) algorithm that leaves $\CB(p,I)$ invariant 
\citep{chen1994weighted,liu1995bayesian}. 
Starting from an arbitrary state $x\in\mathbb{X}$, this MCMC
performs the following steps at each iteration.
\begin{enumerate}
\item Sample $i_{0}\in S_{0}$ and $i_{1} \in S_{1}$ uniformly, independently from one another.
\item Propose to set $x_{i_{0}} = 1$ and $x_{i_{1}} = 0$, and accept with probability 
$\min(1,w_{i_{0}}/w_{i_{1}})$.
\end{enumerate}
By keeping track of the sets $S_0$ and $S_1$, 
the algorithm can be implemented using a constant cost per iteration.
The purpose of this article is to show that this Markov chain
converges to its target distribution $\CB(p,I)$
in the order of $N \log N$ iterations, under mild conditions on $p$ and $I$. 
As the cost per iteration is constant, this provides an overall competitive scheme to 
sample from $\CB(p,I)$.


\subsection{Approach and related works}

We denote the transition kernel of the above Metropolis--Hastings algorithm by $P(x,\cdot)$,
and a Markov chain generated using the algorithm by $(x^{(t)})_{t\geq0}$,
starting from $x^{(0)}\sim\pi_{0}$.  
The initial distribution $\pi_0$ could correspond to setting
$I$ components of $x^{(0)}$ to $1$, chosen uniformly without replacement, or setting
$x_i = 1$ for $i=1,\ldots,I$ and the other components to $0$. 

If all probabilities in $p$ are identical, the chain is equivalent to the
Bernoulli--Laplace diffusion model, which is well-studied
\citep{diaconis1987time,donnelly1994approach,eskenazis2020}.  In particular,
\citet{diaconis1987time} showed that mixing of the chain occurs in the order of
$N\log N$ iterations when $I$ is proportional to $N$, via a Fourier analysis of the
group structure of the chain. A mixing time of the same order can be obtained
with a simple coupling argument \citep{guruswami2000rapidly}.  Here we consider
the case where $p$ is a vector of realizations of random variables in $(0,1)$,
and provide conditions under which the mixing time remains of order $N\log N$.
As we will see in Section \ref{sec:cvgcoupling}, the coupling argument alone
falls apart in the case of unequal probabilities $p$, 
but can be successfully 
combined with a partition of the pair of state spaces into favorable and
unfavorable pairs, to be defined in Section \ref{sec:partition}.
Bounds on the transitions from parts of the space are used to 
define a simple Markov chain on the partition labels, which allows
us to obtain our bounds in Section \ref{sec:mixingtime}.

The problem of sampling $\CB(p,I)$ has 
various applications, such as survey sampling \citep{chen1994weighted},
hypothesis testing in logistic regression \citep{chen1997statistical,brostrom2000acceptance}, 
testing the hypothesis of proportional hazards \citep{brostrom2000acceptance},
and sampling from a determinantal point process 
\citep{hough2006determinantal,kulesza2012determinantal}. 



\section{Convergence rate via couplings \label{sec:cvgcoupling}}

\subsection{General strategy}

Consider two chains $(x^{(t)})$ and $(\tilde{x}^{(t)})$, each marginally
evolving according to $P$, with initialization $x^{(0)}\sim \pi^{(0)}$ and
$\tilde{x}^{(0)}\sim \CB(p,I)$.  Define the sets
$\tilde{S}_{z}=\{n\in[N]:\tilde{x}_{n}=z\}$ for $z=0,1$.  The Hamming distance
between two states $x$ and $\tilde{x}$ is $d(x,\tilde{x}) = \sum_{n=1}^N
\mathds{1}(x_n\neq \tilde{x}_n)$.  Since $x,\tilde{x}\in\mathbb{X}$ sum to $I$,
the distance $d(x,\tilde{x})$ must be an even number.
If $d(x,\tilde{x})=D$ then
$|S_{0}\cap\tilde{S}_{0}|=N-I-D/2$, $|S_{1}\cap\tilde{S}_{1}|=I-D/2$ 
and $|\tilde{S}_{0}\cap S_{1}|=|S_{0}\cap\tilde{S}_{1}|=D/2$, 
where $|\cdot|$ denotes the cardinality of a set. 

Let $d^{(t)}$ denote the distance between $x^{(t)}$ and $\tilde{x}^{(t)}$
at iteration $t$.
Following e.g. \citet[Section 6.2,][]{guruswami2000rapidly}, in the case of
identical probabilities $p=(p_1,\ldots,p_N)$, a path coupling strategy
\citep{bubley1997path} gives an accurate upper bound on the mixing time of the chain.
The strategy is to study the distance $d^{(t)}$ as the
iterations progress. 
Denote the total variation distance between the law of $x^{(t)}$ and its limiting distribution
$\CB(p,I)$ by $\|x^{(t)}-\CB(p,I)\|_{\text{TV}}$.
By the coupling inequality and Markov's inequality,
\begin{align*}
  \|x^{(t)}-\CB(p,I)\|_{\text{TV}} \leq \mathbb{P}\left(x^{(t)}\neq\tilde{x}^{(t)}\right) & =\mathbb{P}\left(d^{(t)}>0\right)
  \leq\mathbb{E}\left[d^{(t)}\right].
\end{align*}
If the contraction
$\mathbb{E}[d^{(t)}|x^{(t-1)}=x, \tilde{x}^{(t-1)}=\tilde{x}]\leq(1-c)d^{(t-1)}$
holds for all $x,\tilde{x}\in\mathbb{X}$
with $c\in(0,1)$, by induction this implies that $\|x^{(t)}-\CB(p,I)\|_{\text{TV}}$ is less than $(1-c)^t \mathbb{E}[d^{(0)}]$. 
Noting that $\mathbb{E}[d^{(0)}]\leq N$ and writing $\kappa = -(\log(1-c))^{-1}>0$, 
an upper bound on the $\epsilon$-mixing time, defined as the first time
$t$ at which $\|x^{(t)}-\CB(p,I)\|_{\text{TV}}\leq \epsilon$,
is given by $\kappa \log(N/\epsilon)$. 
%
Thus we seek a contraction result, with $c$ as large as possible. 

Instead of considering all pairs $(x,\tilde{x})\in\mathbb{X}^2$, 
the path coupling argument allows us to restrict our attention 
to contraction from pairs of adjacent states.
We write the set of adjacent states as
$\adjset = \{(x,\tilde{x})\in\mathbb{X}^2: d(x,\tilde{x})=2\}$;
see Appendix \ref{sec:pathcoupling} for more details on path coupling.

\subsection{Contraction from adjacent states \label{sec:contraction}}

We now introduce a coupling $\bar{P}$ of $P(x,\cdot)$ and $P(\tilde{x},\cdot)$,
for any pair $(x,\tilde{x})\in\mathbb{X}^2$, although we will primarily be interested in 
the case $(x,\tilde{x})\in\adjset$. 
First, sample $i_{0},\tilde{i}_{0}$
from the following maximal coupling of the uniform distributions on $S_{0}$ and 
$\tilde{S}_{0}$:
\begin{enumerate}
\item with probability $|S_{0}\cap\tilde{S}_{0}|/(N-I)$, sample $i_{0}$
uniformly in $S_{0}\cap\tilde{S}_{0}$ and set $\tilde{i}_{0}=i_{0}$,
\item otherwise sample $i_{0}$ uniformly in $S_{0}\setminus\tilde{S}_{0}$
and $\tilde{i}_{0}$ uniformly in $\tilde{S}_{0}\setminus S_{0}$,
independently. 
\end{enumerate}
We then sample $i_{1},\tilde{i}_{1}$ with a similar coupling, independently of the pair $(i_{0},\tilde{i}_{0})$. 
Using these proposed indices, 
swaps are accepted or rejected using a common uniform random number.
These steps define a coupled transition kernel $\bar{P}((x,\tilde{x}),\cdot)$. 

Under $\bar{P}$, the distance between the chains 
can only decrease, so for $(x',\tilde{x}') \sim \bar{P}((x,\tilde{x}),\cdot)$
from $(x,\tilde{x}) \in \adjset$, the distance $d(x',\tilde{x}')$
is either zero or two. We denote the expected 
contraction from $(x,\tilde{x})$ by $c(x,\tilde{x})$, i.e. 
$\mathbb{E}[d(x',\tilde{x}')|x,\tilde{x}] = (1-c(x,\tilde{x})) d(x,\tilde{x})$.
In the case $(x,\tilde{x}) \in \adjset$,
we denote by $a$ the single index at which $x_{a}=0,\tilde{x}_{a}=1$, 
and by $b$ the single index at which $x_{b}=1,\tilde{x}_{b}=0$.
An illustration of such states is in Table \ref{table:adjset}.
Up to a re-labelling of $x$ and $\tilde{x}$, we can 
assume $w_a \leq w_b$. By considering all possibilities when propagating 
$(x,\tilde{x}) \in \adjset$ through $\bar{P}$, we find that
\begin{align}
  \label{eq:contractionrate}
  c(x,\tilde{x}) & =\frac{1}{N-I}\times\frac{1}{I}\times \left[\left|1 - \frac{w_a}{w_b}\right|
  +\sum_{i_{1}\in S_{1}\cap\tilde{S}_{1}}\min\left(1,\frac{w_{a}}{w_{i_{1}}}\right)+\sum_{i_{0}\in S_{0}\cap\tilde{S}_{0}}\min\left(1,\frac{w_{i_{0}}}{w_{b}}\right)\right].
\end{align}
The derivation of \eqref{eq:contractionrate} is in Appendix
\ref{appx:contractionrate}. The next question is whether this contraction rate
$c(x,\tilde{x})$ can be lower bounded by a quantity of order $N^{-1}$; if this
is the case, a mixing time of order $N\log N$ would follow. 

\begin{table}
\[
\begin{array}{cccccccccccc}
 &1&  &   &  & a & b& & & &  & N\\
 x &0 & \ldots & 1 & \ldots & 0 & 1 &\ldots & 1 & 0 & \ldots & 0 \\
 \tilde{x} & 0  & \ldots & 1 & \ldots & 1 & 0 & \ldots & 1 & 0 & \ldots & 0
\end{array}
\]
\caption{Adjacent states $(x,\tilde{x})\in \adjset$.
They differ at indices $a$ and $b$ only, with $x_a=\tilde{x}_b=0$
and $x_b = \tilde{x}_a = 1$. The other components of $x$ and $\tilde{x}$
are identical, and equal to $0$ or $1$.
\label{table:adjset}}
\end{table}

\subsection{Shortcomings}

If the probabilities $p$ are identical, $c(x,\tilde{x})$ simplifies to 
$(N-2)/\{(N-I)I\}$ for all $(x,\tilde{x})\in \adjset $.
Assuming that $I \propto N$, this is of order $N^{-1}$ 
and leads to a mixing time in $N \log N$ \citep{guruswami2000rapidly}.
It follows from \citet{diaconis1987time} that this contraction rate is sharp in its dependency on $N$. 
The same conclusion holds 
in the case where $(p_n)$ are not identical but are bounded away from $0$ and $1$,
i.e. $w_{n}\in [w_{lb},w_{ub}]$ 
with $0<w_{lb}<w_{ub}<\infty$ independent of $N$.
In that case, we obtain the rate $c(x,\tilde{x}) \geq (N-2)/\{(N-I)I\} w_{lb}/w_{ub}$, 
which worsens as the ratio $w_{lb} / w_{ub}$ gets smaller. 

The main difficulty addressed in this article arises when $\min_n p_n$ and
$\max_n p_n$ get arbitrarily close to $0$ and $1$ as $N$ increases.
This scenario is common, for example if $(p_n)$ are independent Uniform$(0,1)$, 
we have $\min_n w_{n} \sim N^{-1}$ and $\max_n w_{n} \sim N$.  
Thus for $w_a = \min_n w_{n}$ and $w_b = \max_n w_{n}$, the contraction in
$\eqref{eq:contractionrate}$ can be of order $N^{-2}$ when $I \propto N$,
which leads to an upper bound on the mixing time of order $N^2 \log N$. 
To set our expectations appropriately, we follow the approach of
\citet{biswas2019estimating} to obtain empirical upper bounds on the mixing
time as $N$ increases.  
Details of the approach, which itself is based on couplings, are given in 
Appendix \ref{appx:empiricalmixing}.  
Figure \ref{fig:mixingtimeNlogN} shows the estimated upper bound on the
mixing time, divided by $N \log N$, as a function of $N$, when $(p_n)$
are generated (once for each value of $N$) from independent Uniform(0,1) and $I$ is set to $N/2$.  The
figure suggests that the mixing time might scale as $N \log N$. 

\begin{figure}[t]
\begin{centering}
    \subfloat[Meeting times $/ N\log N$]{\begin{centering}
      \includegraphics[width=0.45\textwidth]{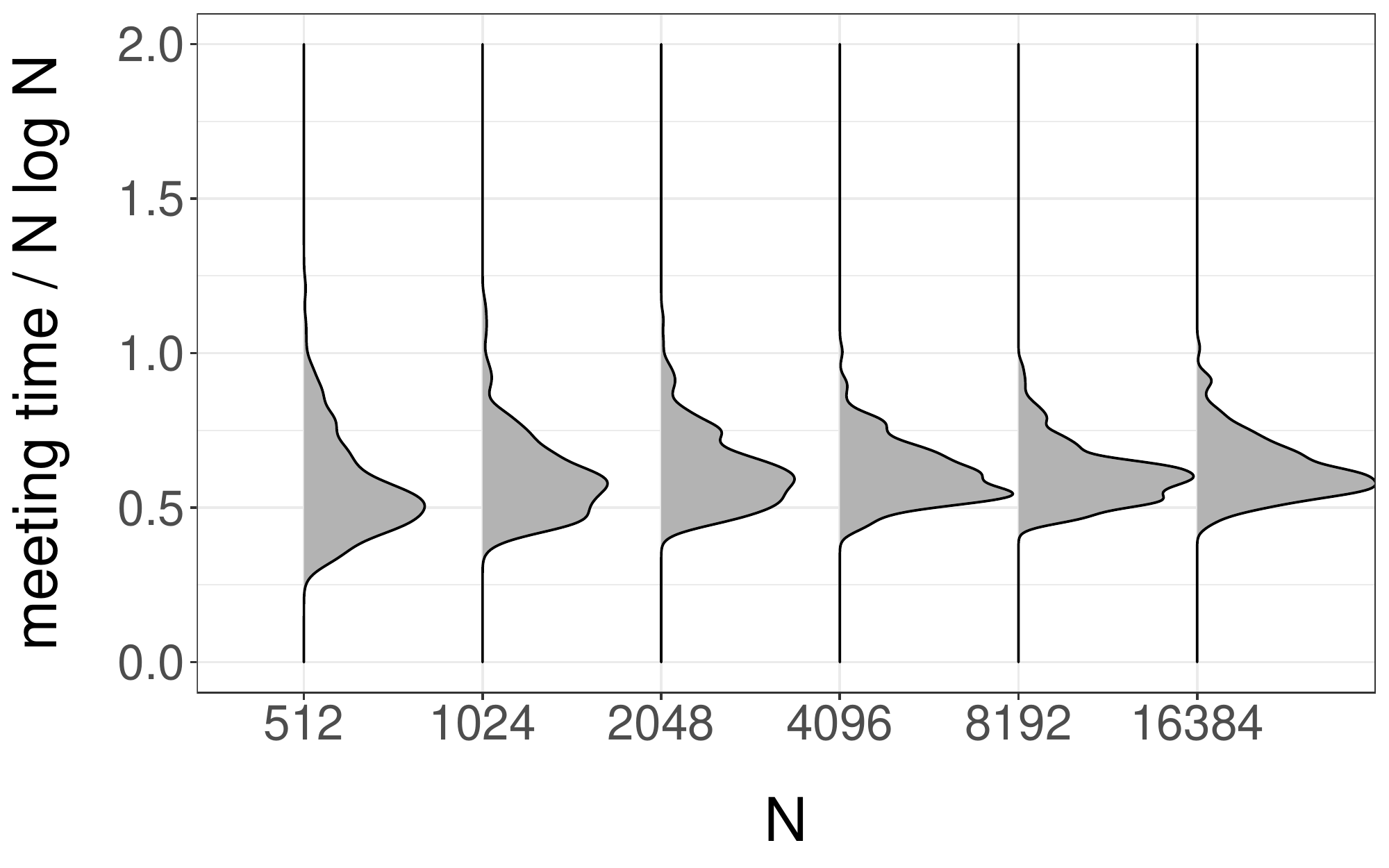}
\par\end{centering}
}
\hspace*{1cm}
\subfloat[Estimated $1\%$-mixing times $/ N\log N$]{\begin{centering}
    \includegraphics[width=0.45\textwidth]{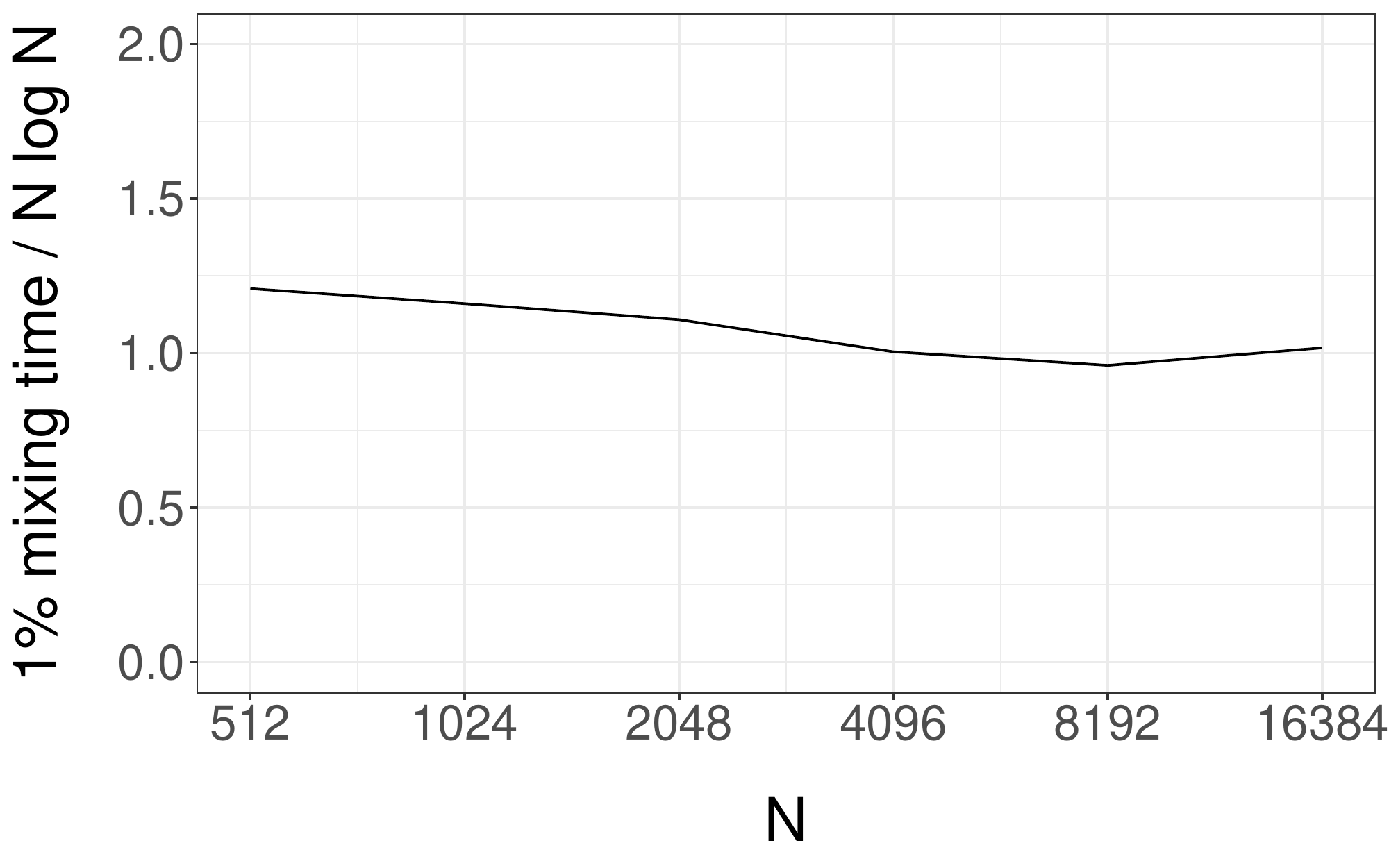}
\par\end{centering}
}
\par\end{centering}
\caption{
  Meeting times (\emph{left}) and estimated upper bounds  on the mixing time of the chain $x^{(t)}$ 
  targeting $\CB(p,I)$ (\emph{right}),
  divided by $N\log N$, against $N$. Here the probabilities $p$ are independent Uniform(0,1)
and $I=N/2$.
  \label{fig:mixingtimeNlogN}
}
\end{figure}

Our contribution is to refine the coupling argument in order to
establish an upper bound on the mixing time of order $N\log N$, under
conditions which allow for example $(p_n)$ to be independent
Uniform(0,1).  A practical consequence of our result stated in Section
\ref{sec:mixingtime} is that the simple MCMC algorithm is competitive compared to
exact sampling strategies for $\CB(p,I)$. 

\section{Proposed analysis}

\subsection{Favorable and unfavorable states \label{sec:partition}}

In the worst case scenario, $w_a$ might be of order $N^{-1}$ and $w_b$ of order $N$,
resulting in a rate $c(x,\tilde{x})$ of order $N^{-2}$. However, this is not
necessarily typical of a pair of states $(x,\tilde{x})\in\adjset$.  
This prompts us to partition $\adjset$ into ``unfavorable'' states, from which their probability
of contracting is smaller than order $N^{-1}$, and ``favorable'' states, from which
meeting occurs with probability of order $N^{-1}$.  
The precise definition of this partition will be made in relation to the odds $(w_n)$. 
Since $c(x,\tilde{x})$ in \eqref{eq:contractionrate} depends on $(w_n)$ and
$I$, we will care about statements holding with high probability under the
distribution of $(w_n)$ and $I$, which are described in Assumptions \ref{def:reasonableodds} and \ref{def:ones}. 
Fortunately, we will see in Proposition \ref{prop:transitionfavorable} that 
favorable states can be reached from unfavorable ones with probability at least 
order $N^{-1}$, while unfavorable states are visited from favorable ones with
probability less than order $N^{-1}$. 
This will prove enough for us to establish a mixing time of order $N\log N$ 
in Theorem \ref{thm:convergencerate}.  

\begin{assumption}\label{def:reasonableodds}
  (Condition on the odds).  The odds $(w_n)$ are such that
  there exist $\zeta>0$, $0<l<r<\infty$ and $\eta>0$ such that
  for all $N$ large enough,
  \[\mathbb{P}\left(\left|\left\{n\in[N]: w_n \notin
  (l, r)\right\}\right|\leq \zeta N\right)\geq
1-\exp(-\eta N).\]
\end{assumption}
This assumption states that with exponentially high probability, a proportion of the odds that falls
within an interval can be defined independently of $N$. The condition can be verified 
using for example Hoeffding's inequality if the odds $(w_n)$ are independently
and identically distributed on $(0,\infty)$, but also under weaker conditions.
The statement ``for all $N$ large enough''
means for all $N\geq N_0$ where $N_0\in\mathbb{N}$.

%

\begin{assumption}\label{def:ones}
  (Conditions on $I$).
  There exist $0<\xi\leq 1/2$ and $\eta' >0$ such that for all $N$ large enough,
  \[\mathbb{P}\left(\xi N \le I \right)\geq 1 - \exp(- \eta' N).\]
\end{assumption}

This assumption formalizes what we mean by $I \propto N$, and is
probabilistic rather than setting $I=\lfloor \xi N\rfloor$ for some $\xi\in
(0,1/2]$. 
It implies that $\xi N^2/2 \le (N-I)I \le (1 - \xi) N^2/2$ with high probability.
Recall that we have assumed $I\leq N/2$ without loss of generality.

\begin{proposition} \label{prop:transitionfavorable}
  Suppose Assumptions \ref{def:reasonableodds} and \ref{def:ones} hold such that $\zeta < \xi$. 
  Then we can define $\xifavtomeet,\xiunfavtofav,\xifavtounfav,\nu>0$ 
  and $0<w_{lo}<w_{hi}<\infty$
  such that, for all $N$ large enough, with probability at least $1-\exp(-\nu N)$,
  the sets of favorable and unfavorable states defined as
\begin{align}
  \unfavset &= \{ (x,\tilde{x})\in \adjset: w_{a}<w_{lo} \text{ and } w_{b}>w_{hi}\},
  \label{eq:unfavorablestates}\\
  \favset &= \{(x,\tilde{x})\in \adjset: w_{a} \ge w_{lo} \ \text{ or } \  w_{b} \le w_{hi}\},
  \label{eq:favorablestates}
\end{align}
and the ``diagonal'' set $ \bar{\mathbb{X}}_{\mathrm{D}} = \{(x,\tilde{x})\in \mathbb{X}^2: x = \tilde{x}\}$, 
satisfy the following statements under the coupling $\bar{P}$ described in Section \ref{sec:contraction},
\begin{align}
  \bar{P}((x,\tilde{x}), \bar{\mathbb{X}}_{\mathrm{D}}) &\geq \xifavtomeet /N,\quad \forall (x,\tilde{x}) \in \favset,
  \label{eq:transit:favtomeet}\\
 \bar{P}((x,\tilde{x}), \favset) &\geq \xiunfavtofav/N, \quad\forall (x,\tilde{x}) \in \unfavset,
 \label{eq:transit:unfavtofav}\\
 \bar{P}((x,\tilde{x}), \unfavset) &\leq \xifavtounfav/N,\quad  \forall (x,\tilde{x}) \in \favset.\label{eq:transit:favtounfav}
\end{align}
\end{proposition}

The proof in Appendix \ref{appx:proofs:transition} relies on a careful inspection of the various cases arising in the
propagation of the coupled chains. 
The proposition provides bounds on the transition 
probabilities between the subsets $\unfavset$, $\favset$ and $\diagset$.

\subsection{Chasing chain and mixing time \label{sec:mixingtime}}

We relate the coupled chain $(x^{(t)},\tilde{x}^{(t)})$ 
to an auxiliary Markov chain denoted by $(Z^{(t)})$,
defined on a space with three states $\{1,2,3\}$, 
associated with the subsets $\unfavset$, $\favset$ and $\diagset$, respectively. 
We introduce the Markov transition matrix 
\begin{align}
  Q = \begin{pmatrix}
    1-\xiunfavtofav/N & \xiunfavtofav/N & 0\\
    \xifavtounfav/N & 1-\xifavtounfav/N-\xifavtomeet/N & \xifavtomeet/N\\
    0 & 0 & 1
  \end{pmatrix},
  \label{eq:transition3states}
\end{align}
where the constants $\xifavtomeet,\xiunfavtofav,\xifavtounfav>0$ are given by Proposition \ref{prop:transitionfavorable}, 
and we assume $N$ is large enough for each entry, including 
$1-\xiunfavtofav/N$ and $ 1-\xifavtounfav/N-\xifavtomeet/N$,  
to be positive.
We then observe that a Markov chain $(Z^{(t)})$ with transition $Q$ 
is such that there exists $r\in(0,1)$ independent of $N$ satisfying $\mathbb{P}(Z^{(N)} = 3 | Z^{(0)} = 1)\geq 1-r$;
details can be found in Appendix \ref{appx:convergenceQ}.

\begin{figure}[t]
  \centering
  \includegraphics[width=0.4\textwidth]{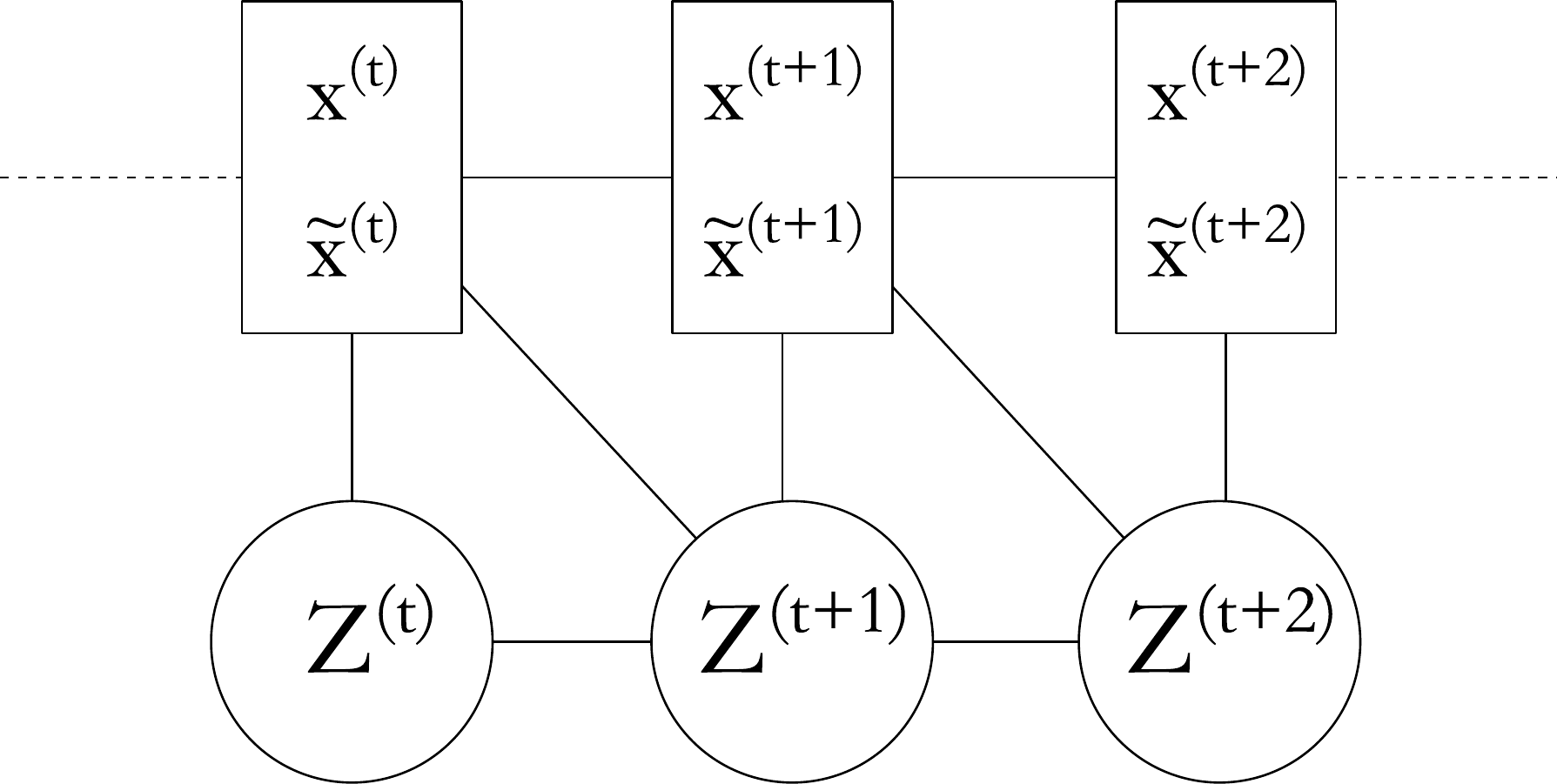}
  \caption{Conditional dependencies of the processes $(x^{(t)},\tilde{x}^{(t)})$ and $(Z^{(t)})$.}
  \label{fig:Zdependencies}
\end{figure}

We now relate the auxiliary chain to $(x^{(t)},\tilde{x}^{(t)})$ using a strategy 
inspired by \citet{jacob2014wang}.
Consider the variable $B^{(t)} \in \{1,2,3\}$ 
defined as $1$ if $(x^{(t)},\tilde{x}^{(t)}) \in \unfavset$,
$2$ if $(x^{(t)},\tilde{x}^{(t)}) \in \favset$ and $3$ if $x^{(t)} = \tilde{x}^{(t)}$.
The key idea is to construct the auxiliary chain $(Z^{(t)})$ on $\{1,2,3\}$, 
in such a way that it is (marginally) a Markov chain with transition matrix
$Q$ in \eqref{eq:transition3states},
and also such that $Z^{(t)}\leq B^{(t)}$ for all $t$ almost surely;
this is possible thanks to Proposition \ref{prop:transitionfavorable}.
Thus the event $\{Z^{(t)} = 3\}$ will imply $\{B^{(t)} = 3\} = \{x^{(t)}=\tilde{x}^{(t)}\}$,
and we can translate the hitting time of $(Z^{(t)})$ 
to its absorbing state
into 
a statement about the meeting time of $(x^{(t)},\tilde{x}^{(t)})$.
An explicit construction of $(Z^{(t)})$ is described in Appendix \ref{appx:constructionauxprocess};
Figure \ref{fig:Zdependencies} represents the dependency structure 
where $Z^{(t+1)}$ is constructed given $Z^{(t)}$, but also conditional
upon $(x^{(t)},\tilde{x}^{(t)})$ and $(x^{(t+1)},\tilde{x}^{(t+1)})$
to ensure that the inequality $Z^{(t+1)}\leq B^{(t+1)}$ holds almost surely.

The convergence of $(Z^{(t)})$ to its
absorbing state translates into
an upper bound on the mixing time of $(x^{(t)})$ of the order of $N\log N$ iterations,
which is our main result.

\begin{theorem} \label{thm:convergencerate}
  Under Assumptions \ref{def:reasonableodds} and \ref{def:ones} such that $\zeta < \xi$, 
  there exist $\kappa>0$, $\nu>0$, $N_0\in\mathbb{N}$ independent of $N$ such that,
  for any $\epsilon\in(0,1)$, and for all $N\geq N_0$,
  with probability at least $1-\exp(-\nu N)$,
  we have 
\[\|x^{(t)} - \CB(p,I)\|_{\mathrm{TV}}\leq \epsilon \quad \text{ for all }\quad t\geq \kappa N \log (N/\epsilon).\]
\end{theorem}

The proof of Theorem \ref{thm:convergencerate} is given in Appendix \ref{appx:upperboundmixing}.

\section{Discussion}

Using the strategy of \citet{biswas2019estimating}, we assess the convergence
rate of the chain in the regime where $I$ is sub-linear in $N$.  Figure
\ref{fig:mixingtimeNlogN_smallI} shows the estimated upper bounds on the mixing
time obtained in the case where $I$ is fixed to $10$ while $N$ grows, and where
$(p_n)$ are independent Uniform(0,1) (generated once for each value of $N$). The figure might suggest that the
mixing time grows at a slower rate than $N$ in this setting,
and thus that MCMC is competitive relative to exact sampling.
Understanding the small $I$ regime remains an open problem. 

\begin{figure}[t]
\begin{centering}
    \subfloat[ Meeting times]{\begin{centering}
      \includegraphics[width=0.45\textwidth]{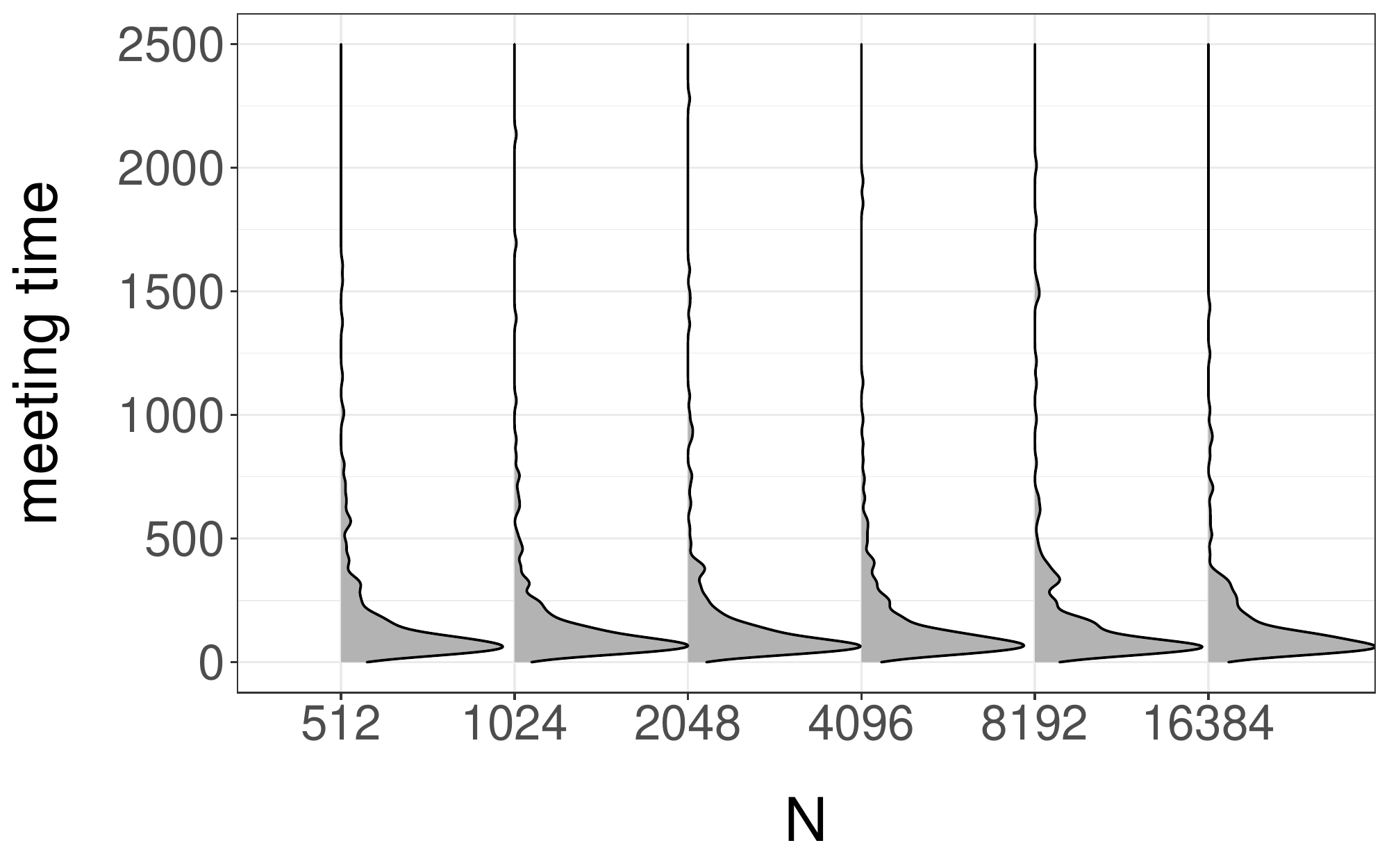}
\par\end{centering}
}
\hspace*{1cm}
\subfloat[Estimated $1\%$-mixing times divided by $N$]{\begin{centering}
    \includegraphics[width=0.45\textwidth]{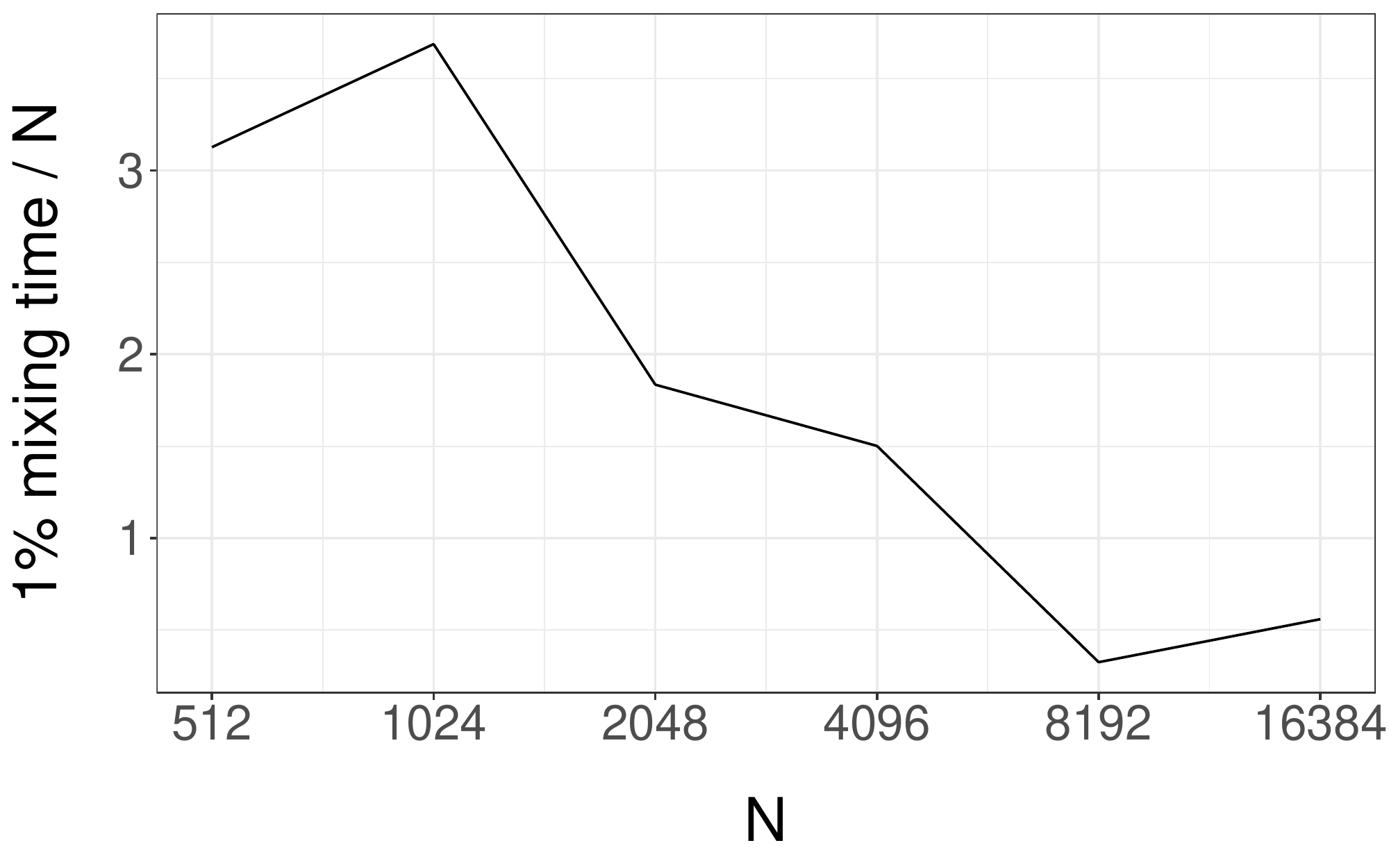}
\par\end{centering}
}
\par\end{centering}
\caption{ Meeting times (\emph{left}) and estimated upper bounds on the mixing time,
  divided by $N$ (\emph{right}), against $N$. The probabilities $p$ are
  independent Uniform(0,1) and $I=10$ for all $N$.
  \label{fig:mixingtimeNlogN_smallI} }
\end{figure}


Our approach relies on a partition of the state space and an auxiliary Markov chain
defined on the subsets given by the partition. 
This technique bear a resemblance to partitioning the state space
with more common drift and contraction conditions 
\citep{durmus2015quantitative,qin2019geometric}, 
but appears to be distinct.

The present setting is similar to the question of sampling permutations via
random swaps. For that problem, direct applications of the coupling argument
result in upper bounds on the mixing time of the order of at least $N^2$.
\citet{bormashenko2011coupling} devises an original variant of the path
coupling strategy to obtain an upper bound in $N \log N$, which is the correct
dependency on $N$; see \citet{berestycki2019cutoff} for recent developments
leading to sharp constants.

The proposed analysis captures the impact of the dimension $N$ faithfully. 
It fails to provide accurate constants and exact characterizations of how the mixing time depends on 
the distribution of the probabilities $(p_n)$ and the sum $I$. 
Yet our analysis already 
supports the use of MCMC over exact sampling strategies
for conditional Bernoulli sampling,
especially as part of encompassing MCMC algorithms
such as that of \citet{yang2016computational} for 
Bayesian variable selection.

%

\subsubsection*{Acknowledgments}
This work was funded by CY Initiative of Excellence (grant ``Investissements
d'Avenir'' ANR-16-IDEX-0008). 
Pierre E. Jacob gratefully acknowledges support
by the National Science Foundation through grants DMS-1712872 and DMS-1844695.

\small 
\bibliographystyle{plainnat}
\bibliography{ref}

\appendix

\section{Exact sampling of conditional Bernoulli \label{appx:exactsampling}}

We describe a procedure to sample exactly from $\CB(p,I)$ for a cost of order $N^2$. First, compute a $(I+1)\times N$ matrix of entries $q(i,n)$,
for $i\in\{0,\ldots,I\},n\in[N]$,
where $q(i,n) = \mathbb{P}(\sum_{m=n}^N x_m = i)$ with each $x_n$ independent Bernoulli$(p_n)$.
To compute these entries, proceed as follows.
The initial conditions are given by  
\begin{equation}
	q(0,n) = \mathbb{P}\left(\sum_{m=n}^N x_m=0\right) = \prod_{m=n}^N\mathbb{P}(x_m=0) = \prod_{m=n}^N (1 - p_m), \quad n\in[N],
\end{equation}
in the case of no success, $q(1,N)=\mathbb{P}(x_N=1)=p_N$
where the sum reduces to a single Bernoulli variable, and 
$q(i,n)=0$ for $i>N-n+1$ because a sum of $N-n+1$ Bernoulli variables
cannot be larger than $N-n+1$, in particular $q(i,N)=0$ for all $i\geq 2$. 
The other entries $q(i,n)$ can be obtained recursively,
via
\[q(i,n) = p_n q(i-1,n+1) + (1 - p_n) q(i,n+1), \quad i\in[N],n\in[N-1].\]
%
Indeed, for $i\in[N]$ and $n\in[N-1]$, by conditioning on the value of
$x_{n}\in\{0,1\}$, the law of total probability gives 
\begin{align}
	q(i,n) &= \mathbb{P}(x_{n}=1)~\mathbb{P}\left(\sum_{m=n}^Nx_m=i \mid x_{n}=1\right) ~+~ 
	\mathbb{P}(x_{n}=0)~\mathbb{P}\left(\sum_{m=n}^Nx_m=i \mid x_{n}=0\right) \notag\\
	&= p_n~\mathbb{P}\left(\sum_{m=n+1}^Nx_m=i-1\right) ~+~ 
	(1-p_n)~\mathbb{P}\left(\sum_{m=n+1}^Nx_m=i \right) .
\end{align}

Having obtained the $(I+1)\times N$ entries $q(i,n)$, 
we now derive a sequential decomposition of a conditioned Bernoulli distribution 
that enables sampling in the order of $N$ operations.
To sample $x_1$, we compute $\mathbb{P}(x_1 = 1 | \sum_{n=1}^N x_n = i)$, as
\begin{align}
  \mathbb{P}\left(x_1 = 1 | \sum_{n=1}^N x_n = i\right) 
  &= \frac{\mathbb{P}\left(x_1 = 1\right)\mathbb{P}\left(\sum_{n=1}^N x_n = i | x_1 = 1\right)}{
    \mathbb{P}\left(\sum_{n=1}^N x_n = i \right)}.
\end{align}
Note that the denominator is $q(i,1)$ 
and the numerator is $p_1 q(i-1,2)$.
%
Similarly for $n\in\{2,\ldots,N-1\}$,
\begin{align}
  \mathbb{P}\left(x_n = 1 |x_{1},\ldots, x_{n-1}, \sum_{n=1}^N x_n = i\right) &=
  \frac{\mathbb{P}\left(x_n = 1\right)\mathbb{P}\left(\sum_{m=n}^N x_m=i-i_{n-1} | x_n = 1\right)}
	{\mathbb{P}\left(\sum_{m=n}^N x_m=i-i_{n-1}\right)},
\end{align}
with $i_n=\sum_{m=1}^n x_m$. The numerator can be recognized as $p_n q(i-i_{n-1}-1,n+1)$
and the denominator as $q(i-i_{n-1},n)$. 
Lastly, given $x_{1},\ldots,x_{N-1},\sum_{n=1}^N x_n =i$, we can set $x_N$ to zero or one
deterministically, namely $x_N = i - i_{N-1}$. 
%

\section{Contractive coupling}

\subsection{Path coupling\label{sec:pathcoupling}}

Given current states $(x,\tilde{x})\in\mathbb{X}^2$, let $(x',\tilde{x}') \sim \bar{P}((x,\tilde{x}),\cdot)$ denote 
new states sampled from $\bar{P}((x,\tilde{x}),\cdot)$, a coupling of $P(x,\cdot)$
and $P(\tilde{x},\cdot)$. 
We want to establish the contraction 
\begin{align}\label{eqn:contraction_allpairs}
	\mathbb{E}\left[d(x',\tilde{x}')|x,\tilde{x}\right]\leq (1-c) d(x,\tilde{x}),	
	\quad \forall (x,\tilde{x})\in\mathbb{X}^2,	
\end{align}
with a large contraction rate $c\in(0,1)$. 
The path coupling argument
\citep{bubley1997path,guruswami2000rapidly} 
allows us to reduce the task in \eqref{eqn:contraction_allpairs} to   
contraction from pairs of adjacent states, i.e. 
\begin{align}\label{eqn:contraction_adjpairs}
	\mathbb{E}\left[d(x',\tilde{x}')|x,\tilde{x}\right]\leq (1-c) d(x,\tilde{x}),	
	\quad \forall (x,\tilde{x})\in \adjset.
\end{align}

It operates as follows. Suppose that \eqref{eqn:contraction_adjpairs} holds 
under the coupling $\bar{P}$. 
For two arbitrary states $(x,\tilde{x})\notin\adjset$ with $d(x,\tilde{x})=D>2$,
we consider a ``path'' $x=z_{0},z_{1},\ldots,z_{L}=\tilde{x}$ of
$L=D/2$ adjacent elements (i.e. $d(z_{\ell},z_{\ell+1})=2$ for $\ell=0,\ldots,L-1$).
By construction the sum $\sum_{\ell=1}^L d(z_{\ell-1},z_{\ell})$ equals $D$.
As there could be multiple such paths, to remove any ambiguity,
we define a deterministic path by going through $x$ and $\tilde{x}$ from left to right,
introducing a new element in the path for each encountered discrepancy. 
We then generate the new states $(x',\tilde{x}')$ using the following procedure:
\begin{enumerate}
	\item sample $(z_0',z_{1}')\sim\bar{P}((z_0,z_{1}),\cdot)$,
	\item for $\ell=2,\ldots,L$, sample $z'_{\ell}$ from the conditional of
$\bar{P}((z_{\ell-1},z_{\ell}), (z'_{\ell-1},z_{\ell}'))$ given $z'_{\ell-1}$, 
	\item set $x' = z'_0$ and $\tilde{x}'=z'_{L}$.
\end{enumerate}
By construction we have $\tilde{x}'|\tilde{x}\sim P(\tilde{x},\cdot)$,
thus this scheme defines a coupling
of $P(x,\cdot)$ and $P(\tilde{x},\cdot)$. 
Under the above coupling, we have 
\begin{align*}
\mathbb{E}\left[d(x',\tilde{x}')|x,\tilde{x}\right] & =\mathbb{E}\left[d(x',\tilde{x}')|x=z_{0},\ldots,z_{L}=\tilde{x}\right]\\
  \leq\mathbb{E}\left[\sum_{\ell=1}^{L}d(z'_{\ell-1},z'_{\ell})|x=z_{0},\ldots,z_{L}=\tilde{x}\right]
 & =\sum_{\ell=1}^{L}\mathbb{E}\left[d(z'_{\ell-1},z'_{\ell})|z{}_{\ell-1},z{}_{\ell}\right]
  \leq\sum_{\ell=1}^{L}(1-c)d(z_{\ell-1},z_{\ell}),
\end{align*}
for any $(x,\tilde{x})\notin\adjset$. 
The first equality holds because $(z_{\ell})$ is obtained deterministically
given $x,\tilde{x}$. The rest follow from triangle
inequalities, linearity of expectation,
conditional independencies between the variables introduced
in the coupling construction, and the assumption of contraction from adjacent
states in \eqref{eqn:contraction_adjpairs}. The last expression is equal to $(1-c)d(x,\tilde{x})$ by construction of the 
path. 
In summary, the path coupling argument allows us to extend contraction 
between adjacent states \eqref{eqn:contraction_adjpairs} to 
contraction for any pair of states \eqref{eqn:contraction_allpairs} with the same rate. 

\subsection{Contraction rate of $\bar{P}$\label{appx:contractionrate} for adjacent states}

We compute the contraction rate $c(x,\tilde{x})$ in \eqref{eq:contractionrate}
under $\bar{P}$, the coupling described in Section \ref{sec:contraction}.

The coupling of $(i_0,\tilde{i}_{0})$ is such that 
$\mathbb{P}(i_0=a,\tilde{i}_0=b) = (N-I)^{-1}$. Similarly the maximal coupling on $(i_1,\tilde{i}_{1})$
leads to $\mathbb{P}(i_1=b,\tilde{i}_1=a) = I^{-1}$.
Under that coupling, if
none of the indices $i_{0},\tilde{i}_{0},i_{1},\tilde{i}_{1}$ are
in $\{a,b\}$, then the proposed swaps will be either accepted or
rejected jointly and the distance $d(x,\tilde{x})$ will be unchanged.
We consider proposed swaps that could affect the discrepancy. 
\begin{enumerate}
\item The index $i_{0}$ can be equal to $a$ (with probability
  $(N-I)^{-1}$). In that case, $\tilde{i}_{0}$ must be equal to $b$,
which is the only index in $\tilde{S}_{0}$ that is not in $S_{0}$;
this comes from the maximal coupling strategy for sampling $(i_0, \tilde{i}_0)$.
The index $i_{1}$ can be equal to $b$ (with probability $I^{-1}$), 
in which case $\tilde{i}_{1}=a$ again due to the maximal coupling
strategy. The discrepancy is then reduced if exactly one of the two proposed swaps
is accepted, which happens with probability 
\[
\left|\min\left(1,\frac{w_{i_{0}}}{w_{i_{1}}}\right)-\min\left(1,\frac{w_{\tilde{i}_{0}}}{w_{\tilde{i}_{1}}}\right)\right|
=|1 - w_a / w_b|.
\]
\item The index $i_{0}$ can be equal to $a$ (again with probability $(N-I)^{-1}$)
  and $i_{1}$ not equal to $b$ (with probability $(I-1)I^{-1}$).
Then $\tilde{i}_{1}=i_{1}$ (maximum coupling), and given $i_{1}\neq b$,
the discrepancy is reduced if both proposed swaps are accepted. 
The probability of reducing the discrepancy given $\{i_{0}=a,i_{1}\neq b$\}
is 
\[
\frac{1}{I-1}\sum_{i_{1}\in S_{1}\cap\tilde{S}_{1}}\min\left(1,\frac{w_{a}}{w_{i_{1}}}\right).
\]
\item The index $i_{0}$ can be different from $a$, in which case $\tilde{i}_{0}=i_{0}$,
and the discrepancy might be reduced if $i_{1}=b$,
$\tilde{i}_{1}=a$ and both swaps are accepted. 
Given $\{i_{0}\neq a,\tilde{i}_{1}=b\}$
this occurs with probability 
\[
\frac{1}{N-I-1}\sum_{i_{0}\in S_{0}\cap\tilde{S}_{0}}\min\left(1,\frac{w_{i_{0}}}{w_{b}}\right).
\]
\end{enumerate}
The contraction rate in \eqref{eq:contractionrate} follows from summing up the above three
possibilities.

\section{Estimation of upper bounds on the mixing time  \label{appx:empiricalmixing}}

We briefly describe the choices made in applying the $L$-lag coupling approach of
\citet{biswas2019estimating}.

For the choice of coupling, we implemented the kernel $\bar{P}$ presented in Section
\ref{sec:contraction}. A careful implementation of the kernel, by keeping track of
the four sets $S_{ij}$ of indices $n$ such that $x_n =i,\tilde{x}_n=j$ for $i,j\in\{0,1\}$,
results in a constant cost per iteration of the coupled chain. 
The chains are initialized by sampling $I$ indices without replacement in binary vectors
of length $N$, and setting these components to one and the others to zero.

We use a lag of $L=1$, and run $500$ independent
runs of coupled lagged chains, to obtain as many realizations of the meeting time $\tau$. 
We employ the key identity in \citet{biswas2019estimating}, 
\[\|x^{(t)}-\CB(p,I)\|_{\text{TV}}\leq \mathbb{E}\left[\max\left(0,\left\lceil \frac{\tau - L - t}{L}\right\rceil\right)\right].\]
The expectation on the right hand side is estimated by
an average of independent copies of the meeting time,
for any desired iteration $t$. As the estimate is itself decreasing in $t$,
we can find the smallest $t$ such that the estimate
is less than $\epsilon=0.01$, and this provides an estimated 
upper bound on the $\epsilon$-mixing time. 

\section{Proofs \label{appx:proofs}}

\subsection{Proof of Proposition \ref{prop:transitionfavorable} \label{appx:proofs:transition}}

Under the assumptions,
we define $0<w_{lo}<w_{hi}<\infty$ and $\delta>0$
such that $w_{lo}+\delta = l$ and $w_{hi}-\delta = r$, with $(l,r)$ as the interval in Assumption
\ref{def:reasonableodds}. The assumption thus guarantees
that with high probability, the number of odds in $(w_n)$ that are outside of 
$(w_{lo}+\delta, w_{hi}-\delta)$ is less than $\zeta N$. 
In particular, the number of odds above $w_{hi}$,
and the number of odds below $w_{lo}$ are both less than $\zeta N$.
The $\delta$ term can be arbitrarily small and is helpful in a calculation
below.

\emph{Proof of \eqref{eq:transit:favtomeet}.} We start with the transition from $\favset$ to $\diagset$.
Assume $(x,\tilde{x})\in \favset$.
For such states, $w_a > w_{lo}$ or $w_b < w_{hi}$,
therefore 
the contraction rate $c(x,\tilde{x})$ \eqref{eq:contractionrate} is at least
\begin{align*}
  &\frac{1}{I(N-I)}\left\{ \sum_{i_1\in S_1\cap \tilde{S}_1} \min(1, w_{lo}/w_{i_1}) \right\} \quad \text{ if $w_a > w_{lo}$,}
\\
&
\frac{1}{I(N-I)}\left\{ 
  \sum_{i_0\in S_0\cap\tilde{S}_0} \min(1, w_{i_0}/w_{hi})\right\} \quad \text{ if $w_b < w_{hi}$}. 
\end{align*}

\begin{itemize}
\item First case ($w_a > w_{lo}$):
  Note that there are $I-1$ indices in $S_1 \cap \tilde{S}_1$.
  Using Assumption \ref{def:ones}, $I$ is at least $\xi N$ with high 
  probability. 
  Using Assumption \ref{def:reasonableodds}, 
  the number of odds in $(w_n)$ above $w_{hi}$ is less than $\zeta N$.
  On the intersection of events, which is not empty if $N$ is large enough,
  among the $I-1$ entries in $S_1 \cap \tilde{S}_1$, 
  there are at least $(\xi - \zeta)N-1$ odds that are smaller than $w_{hi}$.
  The sum $\sum_{i_1\in S_1\cap \tilde{S}_1} \min(1, w_{lo}/w_{i_1})$ is thus 
  larger than $((\xi-\zeta)N-1)w_{lo}/w_{hi}$.
  We obtain the lower bound 
  \[\bar{P}((x,\tilde{x}), \diagset) \geq 2(1-\xi)^{-1}N^{-2} \cdot ((\xi-\zeta) N-1)w_{lo}/w_{hi} .\]
\item Second case ($w_b < w_{hi}$):
  In that case, among the $N-I-1$ indices in $S_0\cap \tilde{S}_0$, 
  under the assumptions there are at least $(1/2-\zeta)N-1$ components of $(w_n)$
  that are larger than $w_{lo}$. Thus 
  the sum $\sum_{i_0\in S_0\cap\tilde{S}_0} \min(1, w_{i_0}/w_{hi})$
  is larger than $((1/2-\zeta)N-1)w_{lo}/w_{hi}$, and we obtain the lower bound
  \[\bar{P}((x,\tilde{x}), \diagset) \geq 
  2(1-\xi)^{-1}N^{-2} \cdot ((1/2-\zeta) N-1)w_{lo}/w_{hi} .\]
\end{itemize}

From the two cases, we obtain for $N$ large enough
a lower bound of the form $\xifavtomeet/N$ for some $\xifavtomeet>0$.

\emph{Proof of \eqref{eq:transit:unfavtofav}.} We next consider the probability of transitioning from $\unfavset$ to
$\favset$. For any $(x,\tilde{x})\in\unfavset$, such transitions occurs in two 
distinct cases. 
\begin{itemize}
  \item We propose swapping component $i_0 = a$ and $i_1 \neq b$, and $\tilde{i}_0 = b$ and $\tilde{i}_1=i_1$,
    such that $w_{i_1} < w_{hi}$, and exactly one of the two swaps is accepted. In that case $b$ becomes $i_1$. 
  \item We propose swapping component $i_1=b$ and $i_0\neq a$, and $\tilde{i}_0 = i_0$ and $\tilde{i}_1 = a$, such
that $w_{i_0} > w_{lo}$, and exactly one of the two swaps is accepted. In that case $a$ becomes $i_0$.
\end{itemize}
Note that the case of swapping components $a$ and $b$ results in an unfavorable state, 
upon relabeling of the states $x$ and $\tilde{x}$ to maintain $w_a \leq w_b$. 
As we only seek a lower bound of $\xiunfavtofav/N$ in \eqref{eq:transit:unfavtofav},
it is sufficient to only consider the first case. 
Since $w_a \leq w_b$, if only one swap is accepted, it must be the one on the $\tilde{x}$
chain; then $b$ becomes $i_1$ and $a$ remains unchanged. 

Selection of $i_0 = a$ and $i_1\neq b$ occurs with probability 
$(N-I)^{-1}\times (I-1)I^{-1}$.
Acceptance of exactly one swap occurs with probability $|\min(1,w_a/w_{i_1}) - \min(1,w_{b}/w_{i_1})|$.
Thus the probability of
moving to $\favset$ via the acceptance of one swap is at least
\begin{align*}
  &\frac{1}{(N-I)I} \sum_{i_1\in S_1 \setminus \{b\}} |\min(1,w_a/w_{i_1}) - \min(1,w_{b}/w_{i_1})| \mathds{1}(w_{i_1} < w_{hi})\\
  &=\frac{1}{(N-I)I} \sum_{i_1\in S_1 \setminus \{b\}} |1 - \min(1,w_a/w_{i_1})| \mathds{1}(w_{i_1} < w_{hi}) \quad 
  \text{because $w_b>w_{hi}$}\\
  &\geq 
  \frac{1}{(N-I)I} \sum_{i_1\in S_1 \setminus \{b\}} |1 - \min(1,w_a/w_{i_1})| \mathds{1}(w_{lo} + \delta < w_{i_1} < w_{hi}) \quad 
  \text{because fewer terms}\\
  &\geq 
  \frac{1}{(N-I)I} \sum_{i_1\in S_1 \setminus \{b\}} |1 - w_a/(w_{lo} + \delta)| \mathds{1}(w_{lo} + \delta < w_{i_1} < w_{hi})\\
  &\geq \frac{1}{(N-I)I} ((\xi-\zeta)N-1) |1 - w_a/(w_{lo} + \delta)|.
\end{align*}
The last inequality holds when $I\geq \xi N$,
and when at most $\zeta N$ entries of $(w_{n})$ are outside of
$(w_{lo}+\delta,w_{hi})$; again we work in the intersection
of the high probability events specified by the assumptions. 

We conclude by noting that, for $w_a < w_{lo}$, we have 
$|1 - w_a/(w_{lo} + \delta)| \geq \delta / (w_{lo}+\delta)$,
which is a constant independent of $N$; this is where the $\delta$ term comes in handy.
Thus for $(x,\tilde{x})\in \unfavset$, we can move to $\favset$ with probability 
\[
\bar{P}((x,\tilde{x}), \favset) \geq
2 (1-\xi)^{-1}N^{-2} \cdot ((\xi-\zeta) N-1) \delta / (w_{lo}+\delta),\]
which is at least $\xiunfavtofav/N$ for some $\xiunfavtofav>0$ as $N$ gets large.

\emph{Proof of \eqref{eq:transit:favtounfav}.}
We finally consider the probability of moving from $\favset$ to $\unfavset$.
As we want an upper bound of this quantity, we have to consider all possible routes 
from $(x,\tilde{x})\in\favset$ to $\unfavset$. 
If the state is such that $w_a > w_{lo}$ and $w_b < w_{hi}$,
then the pair cannot transition to an unfavorable state. In other words, $\bar{P}((x,\tilde{x}), \unfavset)$ can be equal to 
zero.
Transition to an unfavorable state occurs in two distinct cases:
\begin{itemize}
\item if $w_a < w_{lo}$, $w_b < w_{hi}$, and if the swap changes $b$ to some $i_1$ with $w_{i_1}>w_{hi}$;
\item if $w_b > w_{hi}$, $w_a > w_{lo}$, 
and if the swap changes $a$ to some $i_0$ with $w_{i_0}<w_{lo}$.
\end{itemize}
The first case happens
if the drawn indices are $(i_0,\tilde{i}_0) = (a,b)$ and $i_1 = \tilde{i}_1$ in $S_1\setminus\{b\}$ such that $w_{i_1}>w_{hi}$,
and if we accept the swap for the chain $\tilde{x}$ but not for $x$,
which occurs with probability $(w_b - w_a)/w_{i_1}$.
Thus the probability associated with this transition is 
\[\frac{1}{(N-I)I} \left\{\sum_{i_1\in S_1: w_{i_1} > w_{hi}} (w_b-w_a)/w_{i_1} \right\} 
\le \frac{1}{(N-I)I} \sum_{i_1\in S_1} \mathds{1}(w_{i_1}>w_{hi})
\le 2 \xi^{-1} N^{-2}\cdot \zeta N,\]
in the event that the number of odds above $w_{hi}$ is less than $\zeta N$.

The second case occurs if the drawn indices are $i_0 =\tilde{i}_0$ in $S_0\setminus\{a\}$ with
$w_{i_0} < w_{lo}$ and $(i_1,\tilde{i}_1) = (b,a)$,
and if we accept the swap for the chain $\tilde{x}$ but not for $x$,
which occurs with probability $w_{i_0}\cdot(w_a^{-1} - w_b^{-1})\leq 1$.
The probability associated with this transition is at most
$$
	\frac{1}{(N-I)I}\left\{ \sum_{i_0 \in S_0: w_{i_0} < w_{lo}}  w_{i_0}\cdot(w_a^{-1} - w_b^{-1}) \right\} \le \frac{1}{(N-I)I} \sum_{i_0 \in S_0} \mathds{1}(w_{i_0} < w_{lo}) \le 2 \xi^{-1} N^{-2}\cdot \zeta N,
$$
in the event that fewer than $\zeta N$ odds are below $w_{lo}$.
Thus for $(x,\tilde{x})\in\favset$, we can upper bound $\bar{P}((x,\tilde{x}),\unfavset)$ by $\xifavtounfav/N$ for some $\xifavtounfav>0$.

\subsection{Convergence of the Markov chain with transition $Q$ \label{appx:convergenceQ}}

The second largest left-eigenvalue of $Q$ in \eqref{eq:transition3states} is 
\[1 - \frac{1}{2N}\left(\xifavtomeet+\xifavtounfav+\xiunfavtofav - \sqrt{\left(\xifavtomeet+\xifavtounfav+\xiunfavtofav\right)^2 - 4 \xifavtomeet\xiunfavtofav} \right).\]
This is of order $1-\alpha /N$ for a positive constant $\alpha$, for all $\xifavtomeet,\xiunfavtofav,\xifavtounfav>0$.
Using the fact that $(1-\alpha/N)^N$ is less than $\exp(-\alpha)$ for all $\alpha>0$,
we can lower bound the probability 
$\mathbb{P}(Z^{(N)}=3|Z^{(0)}=1)$ by a constant independent of $N$.

\subsection{Construction of the auxiliary chain \label{appx:constructionauxprocess}}

We now detail the construction of the auxiliary chain $(Z^{(t)})$. The construction is done 
conditionally on $(x^{(t)},\tilde{x}^{(t)})$. 
We refer readers to Figure \ref{fig:Zdependencies} 
and recall that $B^{(t)}$ is a deterministic function of $(x^{(t)},\tilde{x}^{(t)})$.
Note that the dependencies shown in Figure \ref{fig:Zdependencies} are a consequence
of the following construction. 

First, if $Z^{(t)}=1$, we construct $Z^{(t+1)}$ as follows.
\begin{itemize}
  \item If  $B^{(t)} = 1$ or $B^{(t)} = 2$, 
    \begin{itemize}
      \item if $B^{(t+1)}=1$ set $Z^{(t+1)}=1$, 
      \item otherwise, set  $Z^{(t+1)} = 2$ with probability $(\xiunfavtofav/N) /
    \mathbb{P}(B^{(t+1)}\in \{2,3\}|x^{(t)},\tilde{x}^{(t)})$, and set $Z^{(t+1)} = 1$ otherwise.
    \end{itemize}
  \item If $B^{(t)} = 3$,  sample $Z^{(t+1)}$ from $\{1,2,3\}$ using the probabilities in the first row of \eqref{eq:transition3states}.
\end{itemize} 

Let us check that the above transition probabilities are well-defined and lie in $[0,1]$.
If $B^{(t)}=1$, 
$(\xiunfavtofav/N) / \mathbb{P}(B^{(t+1)}\in \{2,3\}|x^{(t)},\tilde{x}^{(t)})$
is less than 
$(\xiunfavtofav/N) / \mathbb{P}(B^{(t+1)} = 2|x^{(t)},\tilde{x}^{(t)})$
which is less than one by Proposition \ref{prop:transitionfavorable}, Equation \eqref{eq:transit:unfavtofav}.
If $B^{(t)}=2$, 
$(\xiunfavtofav/N) / \mathbb{P}(B^{(t+1)}\in \{2,3\}|x^{(t)},\tilde{x}^{(t)})$
is less than one if $N$ is large enough, using Proposition \ref{prop:transitionfavorable} again.
Indeed 
\[\mathbb{P}(B^{(t+1)} \in \{2,3\} |x^{(t)},\tilde{x}^{(t)}) = 1-\mathbb{P}(B^{(t+1)} = 1|x^{(t)},\tilde{x}^{(t)})\geq 1-\xifavtounfav/N,\]
by Equation \eqref{eq:transit:favtounfav}, and this is larger than
$\xiunfavtofav/N$ 
if $N$ is large enough, for any
$\xiunfavtofav,\xifavtounfav$.

The goal of this construction is that $Z^{(t+1)} = 2$ only if $B^{(t+1)} \in \{2,3\}$, so that 
$Z^{(t+1)}\leq B^{(t+1)}$ holds almost surely. 
We can compute 
\begin{align*}
  \mathbb{P}(Z^{(t+1)}=2|Z^{(t)}=1, x^{(t)},\tilde{x}^{(t)}) &= 
  \mathbb{P}(Z^{(t+1)}=2,B^{(t+1)}=1|Z^{(t)}=1, x^{(t)},\tilde{x}^{(t)})
 \\
  & + 
  \mathbb{P}(Z^{(t+1)}=2,B^{(t+1)}\in \{2,3\}|Z^{(t)}=1, x^{(t)},\tilde{x}^{(t)}),
\end{align*}
which, using the conditional dependencies implied by the construction
is equal to 
\begin{align*}
  0+&\mathbb{P}(Z^{(t+1)}=2|B^{(t+1)}\in\{2,3\},Z^{(t)}=1,  x^{(t)},\tilde{x}^{(t)} )\times
\mathbb{P}(B^{(t+1)}\in\{2,3\} |x^{(t)},\tilde{x}^{(t)})=\xiunfavtofav/N,
\end{align*}
for all values of $B^{(t)}$ in $\{1,2,3\}$. 
Since the probability $\mathbb{P}(Z^{(t+1)}=2|Z^{(t)}=1, x^{(t)},\tilde{x}^{(t)})$ is the same for all 
$(x^{(t)},\tilde{x}^{(t)})$, we deduce that
$\mathbb{P}(Z^{(t+1)}=2|Z^{(t)}=1)=\xiunfavtofav/N$.

We proceed similarly for the second row of \eqref{eq:transition3states},
assuming $Z^{(t)}=2$. In that case we must have $B^{(t)}\geq 2$.
Consider the following construction.
\begin{itemize}
  \item If $B^{(t)}=2$,
    \begin{itemize}
      \item if $B^{(t+1)} = 1$, set $Z^{(t+1)}=1$,
      \item if $B^{(t+1)} = 2$, sample $Z^{(t+1)}$ from $\{1,2,3\}$
        with probabilities
        \[\left(\frac{\xifavtounfav/N - \mathbb{P}(B^{(t+1)}=1|x^{(t)},\tilde{x}^{(t)})}{
          \mathbb{P}(B^{(t+1)} = 2|x^{(t)},\tilde{x}^{(t)})},
1- \frac{\xifavtounfav/N - \mathbb{P}(B^{(t+1)}=1|x^{(t)},\tilde{x}^{(t)})}{\mathbb{P}(B^{(t+1)} = 2|x^{(t)},\tilde{x}^{(t)})},
0
        \right),\]
      \item if $B^{(t+1)} = 3$, sample $Z^{(t+1)}$ from $\{1,2,3\}$ 
        with probabilities 
        \[\left(0,
1- \frac{\xifavtomeet/N}{\mathbb{P}(B^{(t+1)} = 3|x^{(t)},\tilde{x}^{(t)})},
\frac{\xifavtomeet/N}{\mathbb{P}(B^{(t+1)} = 3|x^{(t)},\tilde{x}^{(t)})}
        \right).\]
    \end{itemize}
  \item If $B^{(t)} = 3$, sample $Z^{(t+1)}$ from $\{1,2,3\}$ using the probabilities in the second row of \eqref{eq:transition3states}. 
\end{itemize}

We can again verify that the probabilities are well-defined and lie in $[0,1]$,
using Proposition \ref{prop:transitionfavorable} and assuming that $N$ is large enough
so that $\mathbb{P}(B^{(t+1)} = 2|x^{(t)},\tilde{x}^{(t)}) + \mathbb{P}(B^{(t+1)} = 1|x^{(t)},\tilde{x}^{(t)}) \ge \xifavtounfav / N $.

Then we can compute 
\begin{align*}
   \mathbb{P}(Z^{(t+1)}=1|Z^{(t)}=2, x^{(t)},\tilde{x}^{(t)}) &= 
\mathbb{P}(Z^{(t+1)}=1,B^{(t+1)}=1|Z^{(t)}=2, x^{(t)},\tilde{x}^{(t)}) \\
&+ \mathbb{P}(Z^{(t+1)}=1,B^{(t+1)} = 2|Z^{(t)}=2, x^{(t)},\tilde{x}^{(t)}) + 0.
\end{align*}
If $B^{(t)}=2$, this becomes
\begin{align*}
& 1\times \mathbb{P}(B^{(t+1)}=1|x^{(t)},\tilde{x}^{(t)}) \\
+& \left(\frac{\xifavtounfav/N - \mathbb{P}(B^{(t+1)}=1|x^{(t)},\tilde{x}^{(t)})}{ \mathbb{P}(B^{(t+1)}=2|x^{(t)},\tilde{x}^{(t)})} \right) 
\times \mathbb{P}(B^{(t+1)}=2|x^{(t)},\tilde{x}^{(t)}) \\
+& 0 \times \mathbb{P}(B^{(t+1)}=3|x^{(t)},\tilde{x}^{(t)}) = \xifavtounfav/N.
\end{align*}
If $B^{(t)}=3$, we also have $\mathbb{P}(Z^{(t+1)}=1|Z^{(t)}=2, x^{(t)},\tilde{x}^{(t)}) = \xifavtounfav/N$.

We next compute the transition from state $2$ to state $3$,
\begin{align*}
   \mathbb{P}(Z^{(t+1)}=3|Z^{(t)}=2, x^{(t)},\tilde{x}^{(t)}) &= 
\mathbb{P}(Z^{(t+1)}=3,B^{(t+1)}=1|Z^{(t)}=2, x^{(t)},\tilde{x}^{(t)}) \\
&\;+\; \mathbb{P}(Z^{(t+1)}=3,B^{(t+1)}=2|Z^{(t)}=2, x^{(t)},\tilde{x}^{(t)}) \\
&\;+\; \mathbb{P}(Z^{(t+1)}=3,B^{(t+1)}=3|Z^{(t)}=2, x^{(t)},\tilde{x}^{(t)}).
\end{align*}
If  $B^{(t)}=2$, this becomes
\begin{align*}
 & 0 \times \mathbb{P}(B^{(t+1)}=1|x^{(t)},\tilde{x}^{(t)})  \\
&\; + \; 0 \times \mathbb{P}(B^{(t+1)}=2|x^{(t)},\tilde{x}^{(t)}) \\
&\; + \;  \frac{\xifavtomeet/N}{\mathbb{P}(B^{(t+1)}=3|x^{(t)},\tilde{x}^{(t)})}  \times \mathbb{P}(B^{(t+1)}=3|x^{(t)},\tilde{x}^{(t)}) = \xifavtomeet/N.
\end{align*}
If $B^{(t)}=3$, we also find that $\mathbb{P}(Z^{(t+1)}=3|Z^{(t)}=2, x^{(t)},\tilde{x}^{(t)}) = \xifavtomeet/N$.
Again these probabilities do not depend on $B^{(t)}$ or $(x^{(t)},\tilde{x}^{(t)})$,
thus the evolution of the chain $(Z^{(t)})$ given $Z^{(t)}=2$ is described
by the second row of \eqref{eq:transition3states}.

\subsection{Upper bound on mixing time \label{appx:upperboundmixing}}

Using the auxiliary chain $(Z^{(t)})$ we 
can state the following result about the $N$-th iteration of $\bar{P}$ from adjacent states.

\begin{proposition} \label{prop:contract}
  Under Assumptions \ref{def:reasonableodds} and \ref{def:ones} 
  such that $\zeta < \xi$, 
  we can define $r\in(0,1)$, $\nu >0$ and $N_0\in\mathbb{N}$ (independent of $N$) such that, for all $N\geq N_0$, 
  with probability at least $1-\exp(-\nu N)$,
  under the coupled kernel $\bar{P}$ 
  \[\mathbb{E}[ d(x^{(N)},\tilde{x}^{(N)})| x^{(0)} = x, \tilde{x}^{(0)} = \tilde{x}] \leq r\, d(x,\tilde{x}) = 2r,\]
  for any two states $(x,\tilde{x})\in\adjset$.
\end{proposition}

\begin{proof}[Proof of Proposition \ref{prop:contract}]
  Under the assumptions we can apply Proposition \ref{prop:transitionfavorable}.
  We place ourselves on the large probability event from that proposition.
  This gives us the constants needed to define the transition matrix $Q$
  in \eqref{eq:transition3states},
  the Markov chain $(Z^{(t)})$ with transition \eqref{eq:transition3states}
  and such that $Z^{(t)}\leq B^{(t)}$ for all $t\geq 0$ (almost surely),
  and the constant $r\in(0,1)$
  such that $\mathbb{P}(Z^{(N)}=3|Z^{(0)}=1)\geq 1-r$.
  Based on the construction, $\mathbb{P}(x^{(N)}=\tilde{x}^{(N)})\geq
  \mathbb{P}(Z^{(N)}=3| Z^{(0)}=1)$,
  thus $\mathbb{P}(x^{(N)}=\tilde{x}^{(N)}| x^{(0)} = x, \tilde{x}^{(0)} = \tilde{x})\geq 1-r$. 
  Noting that $\mathbb{E}[d(x^{(N)}, \tilde{x}^{(N)})| x^{(0)} = x, \tilde{x}^{(0)} = \tilde{x}] = 2\mathbb{P}(x^{(N)}\neq \tilde{x}^{(N)}| x^{(0)} = x, \tilde{x}^{(0)} = \tilde{x})$ concludes the proof.
\end{proof}

We can finally return to the path coupling argument, and apply it to a chain
that follows $P^N$, the $N$-th iterate of the transition kernel of the
original chain.  From Proposition \ref{prop:contract},
we have a contraction rate of $r\in(0,1)$ independently of $N$,
for a coupling of $P^N$
from adjacent states, and we obtain the main theorem as
follows.

\begin{proof}[Proof of Theorem \ref{thm:convergencerate}]
    The path coupling argument shows that 
    for a chain $(\check{x}^{(t)})$ evolving according to $P^N$, 
    there exist $\kappa,\nu >0$ such that for any $\epsilon>0$,
  with probability at least $1-\exp(-\nu N)$,
  we have 
  \[\|\check{x}^{(t)} - \CB(p,I)\|_{\mathrm{TV}}\leq \epsilon \quad\text{ for all }\quad t\geq \kappa \log (N/\epsilon).\]
  The variable $\check{x}^{(t)}$ has the same law as $x^{(tN)}$,
  thus with a change of time variable, $s=tN$, we obtain
  \[\|{x}^{(s)} - \CB(p,I)\|_{\mathrm{TV}}\leq \epsilon \quad\text{ for all }\quad s\geq \kappa N \log (N/\epsilon).\]
\end{proof}

\end{document}